\documentclass{article}

\usepackage[margin=1.5in]{geometry}
\usepackage[utf8]{inputenc} 
\usepackage[T1]{fontenc}    
\usepackage{hyperref}       
\usepackage{url}            
\usepackage{booktabs}       
\usepackage{amsfonts}       
\usepackage{nicefrac}       
\usepackage{microtype}      

\usepackage{graphicx}
\usepackage{amssymb}
\usepackage{lastpage}
\usepackage{transparent}
\usepackage[dvipsnames]{xcolor}  
\usepackage{dsfont}
\usepackage{amsmath}
\usepackage{amsthm}  
\usepackage{lmodern}  
\usepackage{csquotes}  
\usepackage{algorithm}
\usepackage{tabularx}  
\usepackage{ragged2e} 
\usepackage{algpseudocode}
\usepackage{appendix}
\usepackage{booktabs}  
\usepackage{multirow}  

\usepackage[english]{babel}
\usepackage{blindtext}

\newcolumntype{L}[1]{>{\hsize=#1\hsize\RaggedRight} X}

\newtheorem{theorem}{Theorem}[section]
\newtheorem{lemma}[theorem]{Lemma}

\newtheorem{corollary}[theorem]{Corollary}

\theoremstyle{definition}

\newtheorem{definition}{Definition}[section]

\title{Identification of inferential parameters in the covariate-normalized linear conditional logit model}

\author{%
  Philip Erickson\thanks{Author gratefully acknowledges helpful conversations with Florin Bidian, Ziyu Wang, Han Hong, Will Ambrosini, and seminar participants at Amazon.} \\
  Amazon.com\\
  \texttt{phericks@amazon.com} \\
}

\usepackage{biblatex} 
\bibliography{softmax_normalization}

\begin{document}

\maketitle

\begin{abstract}
  The conditional logit model is a standard workhorse approach to estimating customers' product feature preferences using choice data. Using these models at scale, however, can result in numerical imprecision and optimization failure due to a combination of large-valued covariates and the softmax probability function. Standard machine learning approaches alleviate these concerns by applying a normalization scheme to the matrix of covariates, scaling all values to sit within some interval (such as the unit simplex). While this type of normalization is innocuous when using models for prediction, it has the side effect of perturbing the estimated coefficients, which are necessary for researchers interested in inference. This paper shows that, for two common classes of normalizers, designated \emph{scaling} and \emph{centered scaling}, the data-generating non-scaled model parameters can be analytically recovered along with their asymptotic distributions. The paper also shows the numerical performance of the analytical results using an example of a scaling normalizer.
\end{abstract}

\section{Introduction}
The conditional logistic regression, or conditional logit, is a statistical method that models the choice that some type of decision maker makes over a discrete set of potential options. This model takes data on the features of each option as well as the features of the decision maker, combines them with which option the decision maker chose, and allows a researcher to back out the underlying preferences for these features.

Once these preferences are estimated, the model can be used to predict the probability of any of a set of outcomes. For example, if a company were to sell an expensive smart speaker with a high quality sound and an inexpensive smart speaker with a lower-quality sound, what is the probability that customers would buy one or the other? The model can also be used for \emph{inference}, or to study the marginal impact of features on outcome probabilities. For example, a researcher might not only be interested in how many smart speakers a company plans to sell, but also in how adding a screen to a smart speaker might change that number, or how much more impact better sound would have on a customer's willingness-to-pay for a device. These types of inference questions deal not only with final outcomes (predictions), but also the relative impact of different covariates on the outcome.

Because of functional form requirements necessary to estimate conditional logit models, they are prone to numerical issues arising from the scale of model covariates. Large-valued variables, for example, can lead to numerical infinite values that stifle optimization routine convergence. Different scaled variables can also slow down the convergence rate of these optimizers.

The standard machine learning approach to solving numerical problems involving variable scale is normalization, or scaling each variable down to fit within a certain interval (generally either [-1, 1] or [0, 1]). See, for example, \cite{Juszczak_featurescaling} for a use case in support vector machines and \cite{DBLP:journals/corr/IoffeS15} with neural nets. Such approaches solve the numerical infinity problem and boost the efficiency of numerical optimization/solution methods.

While these normalization approaches solve the numerical issues, they have the perverse side-effect of rescaling the model parameter values as well. For most machine learning applications, this is innocuous since the final goal of the model is prediction. However, for inferential problems for which the actual parameter values are used, this rescaling can potentially invalidate the model.

This paper shows that, for two standard classes of normalization schemes (designated scaling and centered scaling) and the standard parameterization of the conditional logit model (linear in parameters), the non-normalized feature parameters can be directly recovered along with their asymptotic distribution. The paper further shows the numerical performance of these results.

The paper proceeds as follows. Section~\ref{sec:model} outlines the linear conditional logit model. Section~\ref{sec:hmin} defines the two classes of normalizers and presents analytical results for recovering data-generating parameters and their asymptotic distribution. Section~\ref{sec:simulation} gives and example showing the numerical performance of the theoretical results using the newly-introduced $varmax$ operator. Section~\ref{sec:conclusion} concludes.

\section{The linear conditional logit model}
\label{sec:model}
The conditional logit model is based on the random utility (RU) framework, developed by McFadden~\cite{McFadden76}. Working from Luce's choice theory~\cite{Luce59}, McFadden provided a method for studying discrete-choice scenarios. The RU approach starts with the assumption that a product gives a person some utility, or well-being, and that when faced with the task of choosing one of several options, a person will choose the option that yields the highest utility.

Formally, the utility of a customer $i$ from product $j$ is given by
\begin{equation}
    u_{ij} = \overline{u}_{ij} + \varepsilon_{ij}
    \label{eq:utility_individual_mixed}
\end{equation}
with $\overline{u}_{ij}$ indicating some fixed level of utility of product $j$ for person $i$ and $\varepsilon_{ij}$ representing the fact that there is noise in customer preferences. Researchers generally assume this random term has a type-1 extreme value (or Gumbel) distribution, which yields an analytical solution to the probability that a customer will purchase any given product. Indexing the purchase event by $t$ and defining $y$ as the purchase decision, the probability of person $i$ choosing product $j$ from a choice set $K$ is given by
\begin{equation}
    \Pr(y_t=j) = \frac{e^{\overline{u}_{ij}}}
                      {1 + \sum_{k\in K}e^{\overline{u}_{ik}}}
    \label{eq:logit_prob}
\end{equation}
which is the standard softmax function. Equation~\eqref{eq:logit_prob} allows for the possibility that a person would choose none of the options in $K$. The probability of this occurring is the same, but with the numerator equal to 1.

While Equation~\eqref{eq:logit_prob} defines a probability as a function of utilities, note that utility itself is latent. The utility can, however, be parameterized as a function of covariates and preference (or utility) parameters, letting the researcher infer the value of the preference parameters by observed choice decisions.

The standard approach to estimating this model is to parameterize the fixed utility as the following linear function of observables
\begin{equation}
    \overline{u}_{ij} = x_{ij}'\beta
    \label{eq:linear_utility}
\end{equation}
with $x_{ij}$ representing a vector of product- and individual-level covariates (and possibly functions of those variables). Equation~\eqref{eq:linear_utility} combined with Equation~\eqref{eq:utility_individual_mixed} characterize the \emph{linear conditional logit model}. Equation~\eqref{eq:linear_utility} also implies that the values of each potential outcome in Equation~\eqref{eq:logit_prob} are now \emph{conditional} on product and customer observable covariates, but not the identity of the products themselves\footnote{This is the primary differentiator between conditional logit models and the multinomial logistic regression.}.

The preference vector $\beta$ in the conditional logit model can be estimated via maximum likelihood by finding the values that most closely justify the outcome data. The likelihood contribution of any purchase decision $t$ is given by
\begin{equation}
    L_t(\beta|y_t, x_t) = \prod_{k' \in K} \Pr(y_t=k')^{\mathds{1}(y_t=k')}
    \label{eq:likelihood_contribution}
\end{equation}
from which the log-likelihood contribution is derived as
\begin{equation}
    \ell_t(\beta|y_t, x_t) = \sum_{k' \in K} \mathds{1}(y_t=k')\log(\Pr(y_t=k'))
    \label{eq:log_likelihood_contribution}
\end{equation}
with the final log-likelihood function given by
\begin{equation}
    \ell(\beta|y, X) = \sum_{t=1}^T \ell_t(\cdot)
    \label{eq:log_likelihood}
\end{equation}
The preference vector $\beta$ is solved for by numerical optimization, finding the values that maximize Equation~\eqref{eq:log_likelihood}, or the likelihood of observing the outcome vector $y$ given the observed covariate matrix $X$.

\section{Normalization and identification}
\label{sec:hmin}
A common problem with approaches using the softmax probability function is variable scale. Since it relies on exponentiation, the softmax function can often produce numbers computationally equivalent to infinity, especially when working with large-valued covariates such as price or income. This can lead to either non-convergence in optimization routines or convergence to nonsensical values.

The standard way to deal with large-valued covariates is through normalization, or scaling the covariates so that they fall within a certain interval (generally [-1, 1] or [0, 1])\cite{han2011data}. These methods are convenient for removing numerical precision issues as in the softmax infinite-value case and often improve the convergence speed of various algorithms as well.

In general, normalization methods are geared towards improving model prediction. In many econometric applications, however, the analyst is interested primarily in the relative impact of different features on the probability of an outcome. That is, $\beta$ is often more important than $\Pr(y)$. And while normalization can clearly improve estimates of $\Pr(y)$, it's not clear that they allow the user to recover an accurate estimate of $\beta$.

One basic class of normalization methods involves Hadamard (or element-wise) multiplying each observation of a covariate matrix by the inverse of some constant vector. Refer to this class as \emph{scaling}. This class is defined formally in Definition~\ref{def:scaling}. The operators $\odot$ and $\oslash$ denote Hadamard multiplication and division, respectively.
\begin{definition}
  Let $x$ be a length-$m$ vector to be normalized and $x_m$ be a length-$m$ normalizing vector with no zero values. Scaling is a normalization method of the form $x \oslash x_m$.
  \label{def:scaling}
\end{definition}
It can be shown that, for any normalizing vector $x_m$, $\beta$ is identified and can be analytically recovered from a model estimated after normalizing the data through scaling. The proof of the result depends on the following Lemma.
\begin{lemma}
    For conforming vectors $x$, $y$, and $z$, $(x \odot y) \cdot z = x \cdot (y \odot z)$
    \label{lem:products}
\end{lemma}
\begin{proof}
    \begin{align*}
        (x \odot y) \cdot z &= [(x_1, \dots, x_n) \odot (y_1, \dots, y_n)] \cdot (z_1, \dots, z_n) \\
            &= (x_1y_1, \dots, x_ny_n) \cdot (z_1, \dots, z_n) \\
            &= x_1y_1z_1 + \dots + x_ny_nz_n \\
            &= (x_1, \dots, x_n) \cdot (y_1z_1, \dots, y_nz_n) \\
            &= (x_1, \dots, x_n) \cdot [(y_1, \dots, y_n) \cdot (z_1, \dots, z_n)] \\
            &= x \cdot (y \odot z)
    \end{align*}
\end{proof}
Using Lemma~\ref{lem:products}, it can now be shown that $\beta$ is identified and can be analytically recovered from an scaling-normalized linear conditional logit model.
\begin{theorem}
    Define $\hat{\beta}^* = \text{argmax } L(\beta^*|X^*)$ and $\hat{\beta} = \text{argmax } L(\beta|X)$ as the MLE estimators for the conditional logit model with an arbitrary row of $X^*$ defined as $x^* = x \oslash x_m$, with $x$ denoting the corresponding row of $X$. Assume that the columns of $X$ are linearly independent. If utility is linear in its parameters, then $\hat{\beta} = \hat{\beta}^*\oslash x_m$.
    \label{thm:identification}
\end{theorem}
\begin{proof}
    It is sufficient to show that for any $x$, $\beta^* = (\beta \odot x_m)$ is the unique solution to $\{\tilde{\beta} \ : \ x'\beta = x^{*\prime}\tilde{\beta}\}$. This can be shown as follows:
    \begin{align*}
        x'\beta &= (x \oslash x_m \odot x_m) \cdot \beta \\
            &= (x^* \odot x_m) \cdot \beta \\
            &= x^* \cdot (x_m \odot \beta) \text{ by Lemma~\ref{lem:products}} \\
            &= x^* \cdot (\beta \odot x_m) \\
            &= x^{*\prime}\beta^*
    \end{align*}

\end{proof}
Since the conditional logit MLE estimate is root-$n$ asymptotically normal, the asymptotic distribution of $\beta^*$ can also be recovered.
\begin{theorem}
    Define $X_m = diag(x_m)$ and let $\Sigma$ denote the asymptotic covariance matrix of $\beta$. The normalized conditional logit estimator has the asymptotic distribution $\sqrt{n}(\hat{\beta}^* - \beta^*) \overset{d}{\longrightarrow} N(0, X_m \Sigma X_m)$.
    \label{thm:distribution}
\end{theorem}
\begin{proof}
    By the delta method, $\sqrt{n}(\hat{\beta}^* - \beta^*) \overset{d}{\longrightarrow} N(0, \nabla \beta^{*\prime} \Sigma \nabla \beta^{*})$, with $\nabla \beta^{*} \equiv \partial \beta^* / \partial \beta$. Since $X_m$ is symmetric, it is sufficient to show that $\nabla \beta^{*} = X_m$. This follows directly from the following matrix reformulation of $\beta^*$
    \begin{equation*}
        \beta^* = \beta \odot x_m = \beta X_m
    \end{equation*}
\end{proof}
The results in Theorems~\ref{thm:identification} and~\ref{thm:distribution} can be extended to other normalization methods as well. A common extension to scaling is \emph{centered scaling}, as defined in Definition~\ref{def:cscaling}.

\begin{definition}
  Let $x$ be a length-$m$ vector to be normalized and $x_m$ and $a$ be two length-$m$ normalizing vectors. Centered scaling is a normalization method of the form $(x - a) \oslash x_m$.
  \label{def:cscaling}
\end{definition}
Centered scaling encompasses several common normalization methods, including min-max and z-score normalization, and is used to additionally decrease estimation sensitivity to outlier observations. Note that Definition~\ref{def:cscaling} only applies to covariates. That is, from a notational standpoint, the intercept must be modeled explicitly rather than allowing one of the columns in $X$ to be 1. Otherwise, Definition~\ref{def:cscaling} would imply that the normalized model doesn't use an intercept. This distinction is irrelevant for scaling. As such, restate the utility function from Equation~\eqref{eq:linear_utility} as

\begin{equation}
    \overline{u}_{ij} = \beta_0 + x_{ij}'\gamma
    \label{eq:linear_utility_intercept}
\end{equation}

\begin{corollary}
  Define $(\hat{\beta}_0^*, \hat{\gamma}*) = \text{argmax } L(\beta_0^*, \gamma|X^*)$ and $(\hat{\beta}_0, \hat{\gamma}) = \text{argmax } L(\beta_0, \gamma|X)$ as the MLE estimators for the conditional logit model with an arbitrary row of $X^*$ defined as $x^* = (x - a) \oslash x_m$, with $x$ denoting the corresponding row of the covariate matrix $X$. Assume that the columns of $X$ and a column of ones are all linearly independent. If utility is linear in its parameters, then $\hat{\gamma} = \hat{\gamma}^*\oslash x_m$ and $\hat{\beta}_0 = \hat{\beta}^*_0 - a'\hat{\gamma}$.
  \label{cor:identification}
\end{corollary}
\begin{proof}
  Following the same procedure as the proof for Theorem~\ref{thm:identification}, it is sufficient to show that for any $x$, $(\beta_0^*, \gamma^*) = (\beta_0 + a'\gamma, \gamma \odot x_m)$ is the unique solution to $\{(\tilde{\beta_0}, \tilde{\gamma}) \ : \ \beta_0 + x'\gamma = \tilde{\beta}_0 + x^{*\prime}\tilde{\gamma}\}$
  \begin{align*}
        \beta_0 + x'\gamma &= \beta_0 + [(x - a + a) \oslash x_m \odot x_m] \cdot \gamma \\
            &= \beta_0 + [(x - a) \oslash x_m \odot x_m + a] \cdot \gamma \\
            &= \beta_0 + [(x - a) \oslash x_m \odot x_m] \cdot \gamma + a \cdot \gamma \\
            &= \beta_0 + [(x - a) \oslash x_m] \cdot (x_m \odot \gamma) + a \cdot \gamma \text{ by Lemma~\ref{lem:products}} \\
            &= \beta_0^* + x^{*\prime}\gamma^*, \text{ since $a'\gamma$ is a constant}\\
    \end{align*}

\end{proof}
Using Corollary~\ref{cor:identification}, recovering the asymptotic distribution of $(\beta_0, \gamma)$ is a straight-forward extension of Theorem~\ref{thm:distribution}.
\begin{corollary}
  Define $\tilde{x}_m$ as $(1, x_m)$, $\tilde{X}_m = diag(\tilde{x}_m)$, and let $\Sigma$ denote the asymptotic covariance matrix of $\beta$, with $\beta=(\beta_0, \gamma)$. The center scaled normalized conditional logit estimator has the asymptotic distribution $\sqrt{n}(\hat{\beta}^* - \beta^*) \overset{d}{\longrightarrow} N(0, X_m \Sigma X_m)$.
    \label{cor:distribution}
\end{corollary}
\begin{proof}
  The proof is the same as for Theorem~\ref{thm:distribution} after noting that, rewriting the result of Corollary~\ref{cor:identification}, $\beta^* = \beta \odot x_m + k$, with $k=(a'\gamma, 0, \dots, 0)$, which is a constant. Therefore, $\nabla \beta^* = X_m$. 
\end{proof}

\section{Simulation performance}
\label{sec:simulation}
This section tests the numerical performance of the paper's analytical results by simulating a set of purchase decisions by customers whose preferences are represented by Equations~\eqref{eq:utility_individual_mixed} and~\eqref{eq:linear_utility}. Consider an example of a scaling using the $varmax$ operator as defined in Definition~\ref{def:varmax}\footnote{This example focuses specifically on scaling normalization, but the same exercise could be replicated for centered scaling.}.

\begin{definition}
    Define the variable maximum operator, or $varmax$, over an $N \times K$ matrix as $varmax(X) = (\max(|x_1|), \dots, \max(|x_K|))$, with $x_k$ denoting the $k$th column of $X$ and $\max(|x_k|)$ the maximum absolute value of $x_k$. Denote $varmax(x)$ as $x_{max}$.
    \label{def:varmax}
\end{definition}
Using $x_{max}$ as the scaling normalizing vector, the covariate matrix will be scaled such that all values fall in the interval [-1, 1].

The simulated model follows the simple linear specification
\begin{equation}
  u_{ij} = x_j'\beta + \varepsilon_{ij}
\end{equation}
with $x_j$ containing four variables (three randomly distributed covariates and an intercept). Simulation details are given in Table~\ref{tab:sim_details}. Note that the covariates are simulated to have drastically different scales and also to have values large enough to generate numerical infinite values in Equation~\eqref{eq:logit_prob}.

\begin{table}[H]
  \begin{center}
    \caption{Simulation details}
    \label{tab:sim_details}
    \begin{tabularx}{\textwidth}{L{1}L{1}}
      \textbf{Item} & \textbf{Value} \\
      \hline \hline
      Number of simulations & 10000 \\
      Number of customers & 2000 \\
      Number of tasks & 25 \\
      Number of options per task & 5 \\
      \hline
      $x_1$ & Intercept (all ones) \\
      $x_2$ & $\sim Binomial(p=0.5)$ \\
      $x_3$ & $\sim Lognormal(\mu=0, \sigma=1)$ \\
      $x_4$ & $\sim Normal(\mu=0, \sigma=5000)$ \\
      $(\beta_1, \beta_2, \beta_3, \beta_4)$ & $(-3, 4, -1.7, 0.00006)$ \\
      \hline
    \end{tabularx}
  \end{center}
\end{table}

Results of the simulation exercise are in Figure~\ref{fig:varmax_sim_perf}, which shows the estimated density functions from the empirical distribution of $\hat{\beta}$ associated with each variable, denormalized according to Theorem~\ref{thm:identification}. Note that without normalization each simulation resulted in numerical infinite values and so failed to converge past an initial guess.
\begin{figure}[H]
  \includegraphics[width=\textwidth]{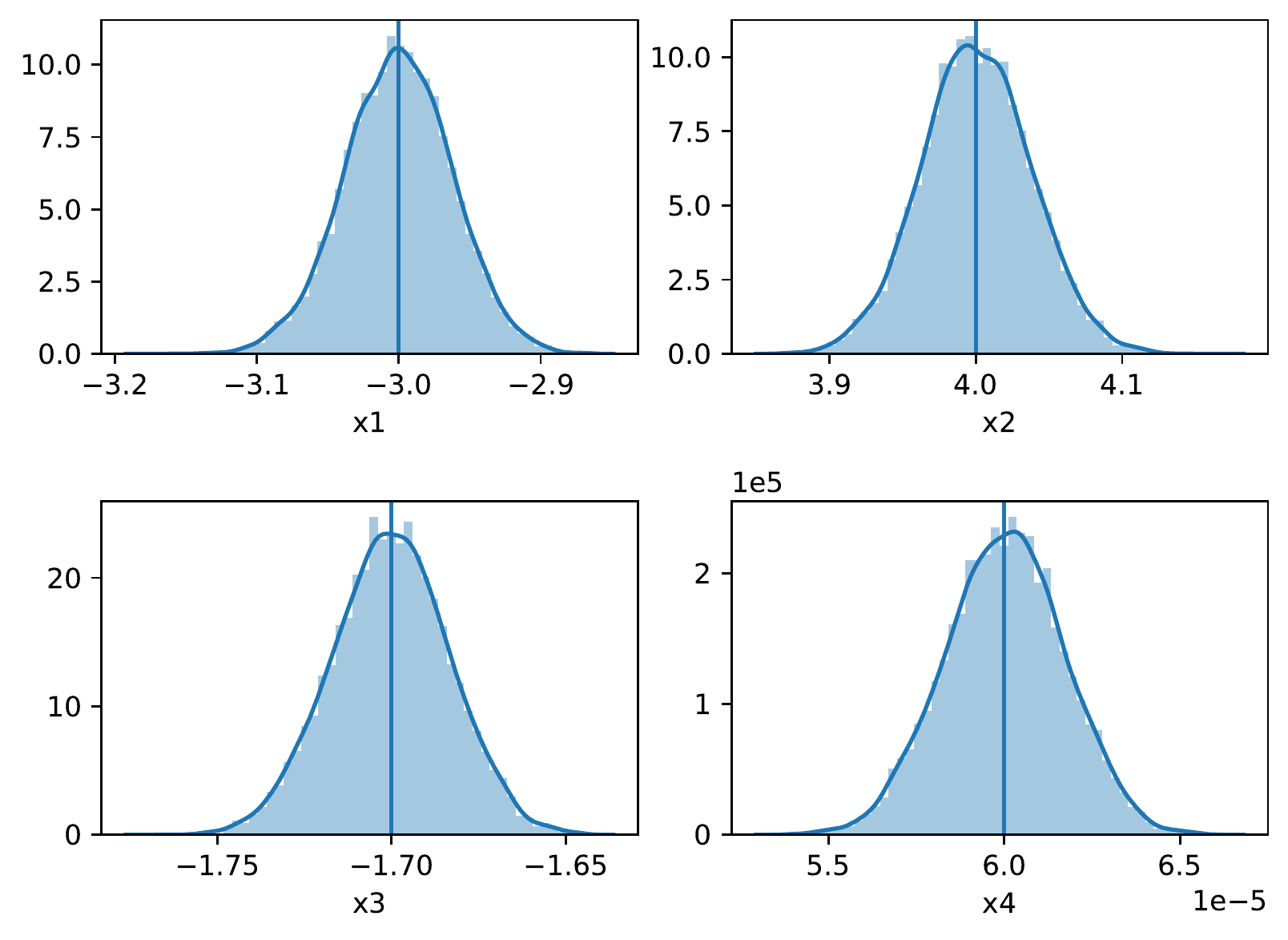}
  \caption{Varmax normalization simulation performance}
  \label{fig:varmax_sim_perf}
\end{figure}
For each variable in $u_{ij}$, the simulated estimates of its corresponding utility weight $\beta$ are centered at the data-generating value with a well-behaved symmetrical distribution, which is exactly what Theorem~\ref{thm:identification} suggests. Further, when lowering the scale of $x_4$ so that non-normalized models can be estimated as well ($\sigma = 5$), the distributions of $\hat{\beta}$ between the normalized and non-normalized models are equivalent according to a 2-sample Kolmogorov-Smirnov test ($p\approx 1$ for all $\beta$).
\section{Conclusion}
\label{sec:conclusion}
Normalization is a powerful tool for improving the numerical efficiency and accuracy of statistical models used for prediction. The results presented here suggest, however, that under certain circumstances, researchers can use normalization for inferential problems as well.

\printbibliography
\end{document}